\documentclass[sigconf]{aamas}  
\usepackage{balance}  

\usepackage{booktabs}

\setcopyright{ifaamas}  
\acmDOI{}  
\acmISBN{}  
\acmConference[AAMAS'19]{Proc.\@ of the 18th International Conference on Autonomous Agents and Multiagent Systems (AAMAS 2019)}{May 13--17, 2019}{Montreal, Canada}{N.~Agmon, M.~E.~Taylor, E.~Elkind, M.~Veloso (eds.)}  
\acmYear{2019}  
\copyrightyear{2019}  
\acmPrice{}  

\begin{document}

\title{PT-ISABB: A Hybrid Tree-based Complete Algorithm to Solve Asymmetric Distributed Constraint Optimization Problems}  


\author{Yanchen Deng}
\affiliation{
	\institution{College of Computer Science, Chongqing University}
	\city{Chongqing} 
	\state{China} 
	\postcode{400044}
}
\email{dyc941126@126.com}
\author{Ziyu Chen}\authornote{Corresponding author.}
\affiliation{
	\institution{College of Computer Science, Chongqing University}
	\city{Chongqing} 
	\state{China} 
	\postcode{400044}
}
\email{chenziyu@cqu.edu.cn}
\author{Dingding Chen}
\affiliation{
	\institution{College of Computer Science, Chongqing University}
	\city{Chongqing} 
	\state{China} 
	\postcode{400044}
}
\email{dingding@cqu.edu.cn}
\author{Xingqiong Jiang}
\affiliation{
	\institution{College of Computer Science, Chongqing University}
	\city{Chongqing} 
	\state{China} 
	\postcode{400044}
}
\email{jxq@cqu.edu.cn}

\author{Qiang Li}
\affiliation{ 
	\institution{College of Electrical Engineering, Chongqing University}
	\city{Chongqing} 
	\state{China} 
	\postcode{400044}
}
\email{qiangli.ac@gmail.com}

\begin{abstract}  
Asymmetric Distributed Constraint Optimization Problems (ADCOPs) have emerged as an important formalism in multi-agent community due to their ability to capture personal preferences.  However, the existing search-based complete algorithms for ADCOPs can only use local knowledge to compute lower bounds, which leads to inefficient pruning and prohibits them from solving large scale problems. On the other hand, inference-based complete algorithms (e.g., DPOP) for Distributed Constraint Optimization Problems (DCOPs) require only a linear number of messages, but they cannot be directly applied into ADCOPs due to a privacy concern. Therefore, in the paper, we consider the possibility of combining inference and search to effectively solve ADCOPs at an acceptable loss of privacy. Specifically, we propose a hybrid complete algorithm called PT-ISABB which uses a tailored inference algorithm to provide tight lower bounds and a tree-based complete search algorithm to exhaust the search space. We prove the correctness of our algorithm and the experimental results demonstrate its superiority over other state-of-the-art complete algorithms.
\end{abstract}

%


\maketitle


\section{Introduction}
Distributed Constraint Optimization Problems (DCOPs) \cite{yeoh2012distributed} are a fundamental framework in multi-agent systems where agents cooperate with each other to optimize a global objective. DCOPs have been successfully deployed in many real world applications including smart grids \cite{fioretto2017distributed}, radio frequency allocation \cite{monteiro2012multi}, task scheduling \cite{sultanik2007modeling}, etc.

Algorithms for DCOPs can generally be classified into two categories, i.e., complete algorithms and incomplete algorithms. Search-based complete algorithms like SBB \cite{hirayamay97}, AFB \cite{gershman2009asynchronous}, ConFB \cite{Netzer2012}, ADOPT \cite{modi2005adopt} and its variants \cite{yeoh2010bnb,gutierrezMY11} perform distributed searches to exhaust the search space, while inference-based complete algorithms including Action-GDL \cite{vinyals2009generalizing}, DPOP \cite{petcuF05} and its variants \cite{petcu2006odpop,petcu2007mb} use dynamic programming to optimally solve problems. In contrast, incomplete algorithms including local search \cite{Zhang2005Distributed,maheswaran2004distributed,okamoto2016distributed}, GDL-based algorithms \cite{farinelli2008decentralised,rogers2011bounded,zivan2012max,chen2017iterative} and sampling-based algorithms \cite{ottens2017duct,fioretto2016dynamic} trade optimality for small computational efforts.

Asymmetric Distributed Constraint Optimization Problems	(ADCOPs) \cite{grinshpounGZNM13} are a notable extension to DCOPs, which can capture ubiquitous asymmetric structures in real world scenarios \cite{maheswaranTBPV04,ramchurn2011agent,burke2007supply}. That is, a constraint in an ADCOP explicitly defines the exact payoff for each participant instead of assuming equal payoffs for constrained agents. Solving ADCOPs is more challenging since algorithms must evaluate and aggregate the payoff for each participant of a constraint. ATWB and SABB \cite{grinshpounGZNM13} are asymmetric versions of AFB and SBB based on an one-phase strategy in which the algorithms systematically check each side of the constraints before reaching a full assignment. Besides, AsymPT-FB \cite{litov2017forward} is another search-based complete algorithm for ADCOPs, which implements a variation of forward bounding on a pseudo tree. However, to the best of our knowledge, there is no asymmetric adaptation of inference-based complete algorithms for DCOPs (e.g., DPOP). That is partially because these algorithms require the total knowledge of each constraint to perform variable elimination optimally. In other words, parent agents must surrender their private constraints to eliminate their children variables, which is unacceptable in a asymmetric scenario.

In this paper, we investigate the possibility of combining both inference and search to efficiently solve ADCOPs at an acceptable loss of privacy. Specifically, our main contributions are listed as follows. 
\begin{itemize}
	\item We propose a hybrid tree-based complete algorithm for ADCOPs, called PT-ISABB.\footnote{The source code is available at https://github.com/czy920/DCOPSovlerAlgorithm\_PTISABB.} The algorithm first uses a tailored version of ADPOP \cite{petcu2005approximations} to solve a subset of constraints, and the inference results stored in agents are served as look-up tables for tight lower bounds. Then, a variant of SABB is implemented on a pseudo tree to guarantee optimality.
	\item We theoretically show the completeness of our proposed algorithm. Moreover, we also prove that the lower bounds in PT-ISABB are at least as tight as the ones in AsymPT-FB when its maximal dimension limit $k=\infty$ .
	\item We empirically evaluate our algorithm on various benchmarks. Our study shows that PT-ISABB requires significantly fewer messages and lower NCLOs than state-of-the-art search-based complete algorithms including AsymPT-FB. The experimental results also indicate that our proposed algorithm leaks less privacy than AsymPT-FB when solving complex problems.
\end{itemize}
\section{Background}
In this section, we review the preliminaries including ADCOPs, pseudo tree, DPOP and ADPOP.
\subsection{Asymmetric Distributed Constraint Optimization Problems}
An asymmetric distributed constraint optimization problem can be defined by a tuple $\langle A,X,D,F\rangle$ in which
\begin{itemize}
	\item $A=\{a_1,\dots,a_q\}$ is a set of agents
	\item $X=\{x_1,\dots,x_n\}$ is a set of variables
	\item $D=\{D_1,\dots,D_n\}$ is a set of finite and discrete domains. Each variable $x_i$ takes a value from $D_i$
	\item $F=\{f_1,\dots,f_m\}$ is a set of constraints. Each constraint $f_i:D_{i_1}\times\dots\times D_{i_k}\rightarrow \mathbb{R}^k_+$ defines a set of non-negative costs for every possible value combination of the set of variables it is involved in
\end{itemize}

For the sake of simplicity, we assume that an agent only controls a variable and all constraints are binary. Therefore, the term \textit{agent} and \textit{variable} can be used interchangeably. Besides, for the constraint between $x_i$ and $x_j$, we denote the private cost functions for $x_i$ and $x_j$ as $f_{ij}$ and $f_{ji}$, respectively. Note that in the asymmetric setting, $f_{ij}$ does not necessarily equal to $f_{ji}$. A solution to an ADCOP is the assignments to all the variables with the minimal aggregated cost. An ADCOP can be visualized by a \textit{constraint graph} in which the vertexes denote the variables and the edges denote the constraints between agents. Fig. 1 (a) visualizes an ADCOP with four agents and four constraints.
\begin{figure}
	\begin{minipage}{0.45\linewidth}
		\centering
		\includegraphics[scale=0.13]{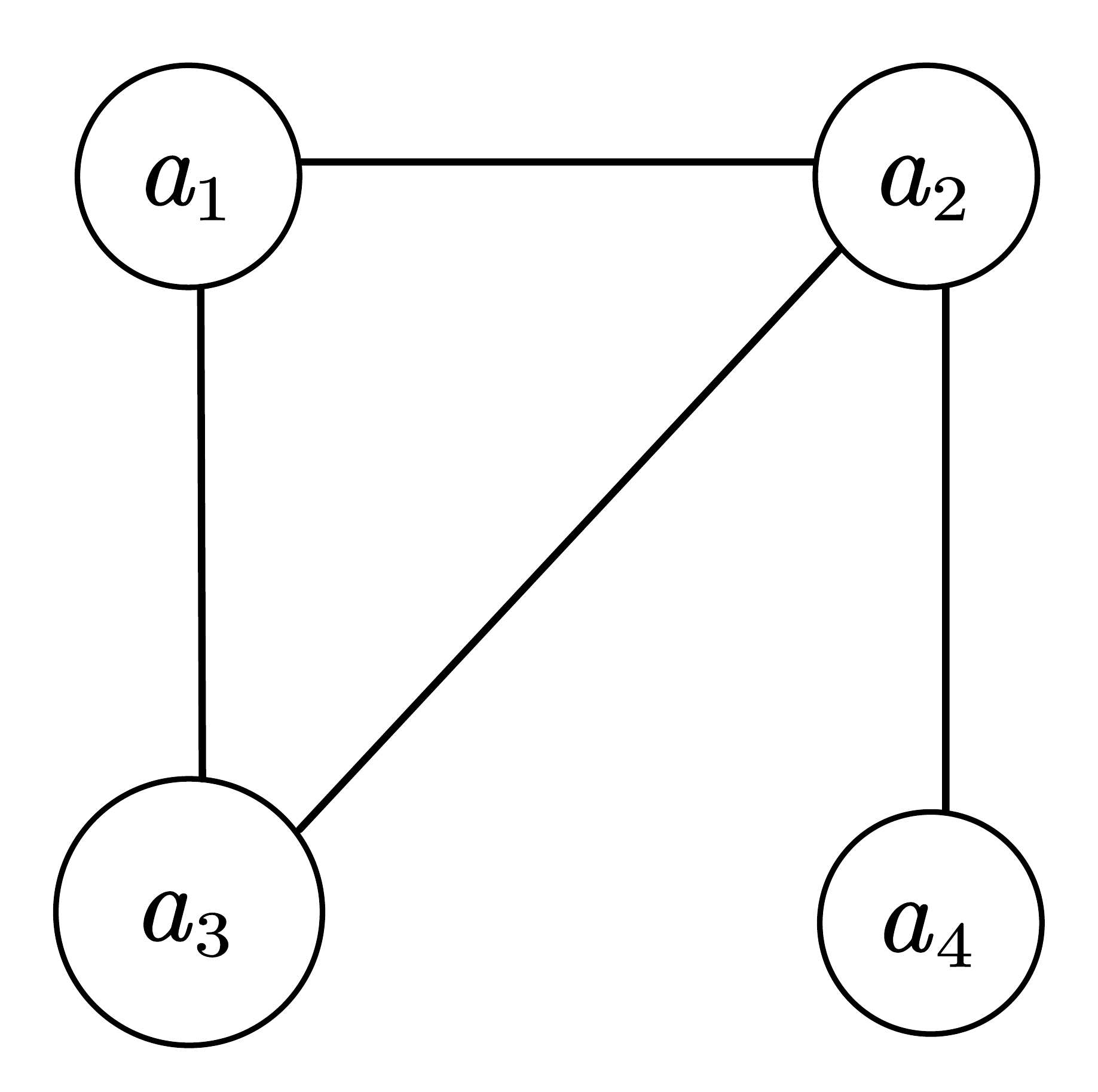}\\
		(a)
	\end{minipage}
	\begin{minipage}{0.45\linewidth}
		\centering
		\includegraphics[scale=0.13]{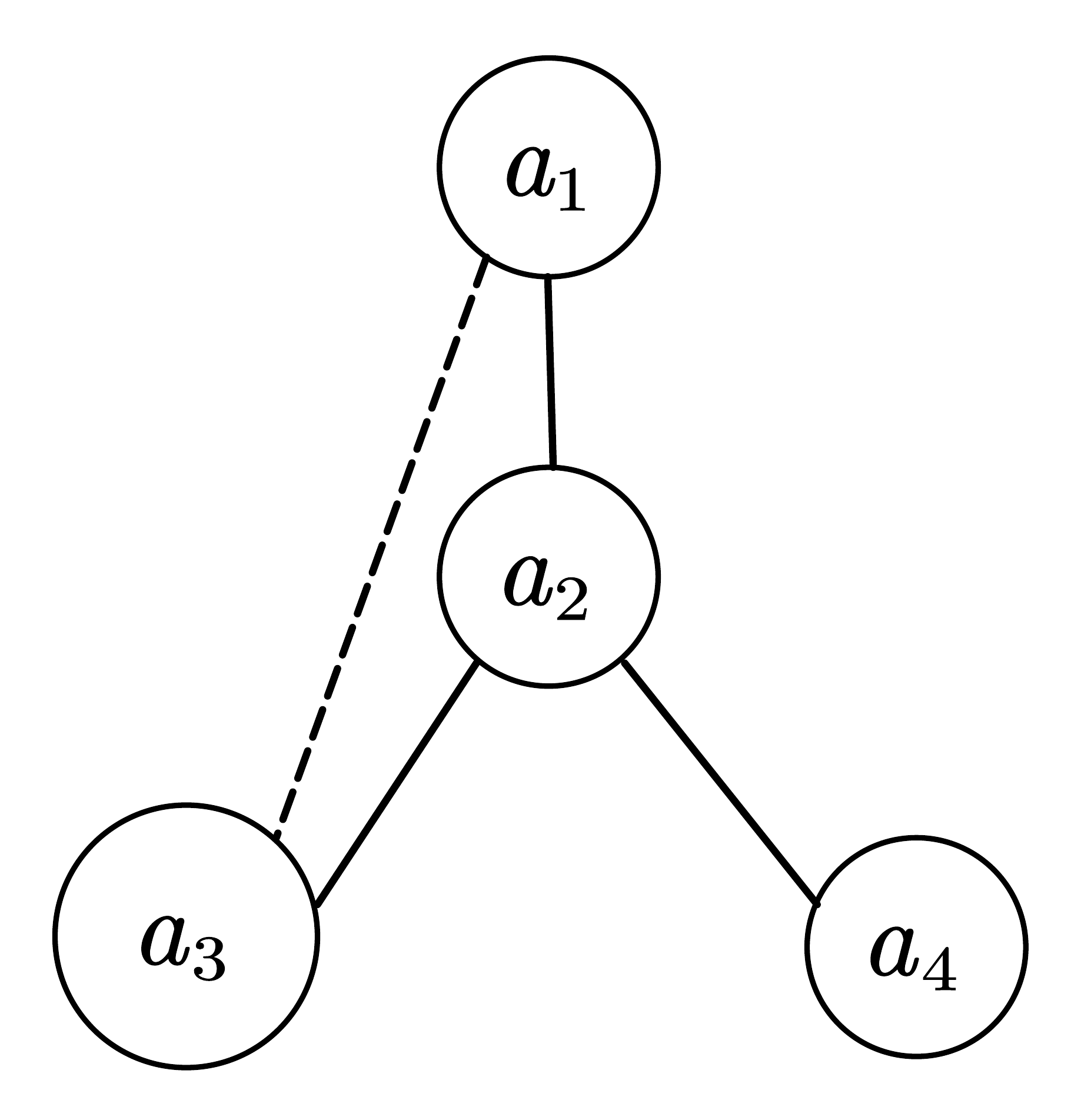}\\
		(b)
	\end{minipage}
	\caption{An example of constraint graph and pseudo tree}
\end{figure}
\subsection{Pseudo Tree}
A pseudo tree is an ordered arrangement to a constraint graph in which agents in different branches are independent and thus search can be performed in parallel on these independent branches. A pseudo tree can be generated by a depth-first traverse of the constraint graph, which categorizes the constraints into tree edges and pseudo edges. For an agent $a_i$, we denote its parent as $P(a_i)$ which is the ancestor connecting to $a_i$ through a tree edge, its pseudo parents as $PP(a_i)$ which is a set of ancestors connecting to $a_i$ through pseudo edges, its children and pseudo children as $C(a_i)$ and $PC(a_i)$ which are the sets of descendants connecting to $a_i$ via tree edges and pseudo edges, respectively. For the sake of clarity, we denote all the parents of $a_i$ as $AP(a_i)=PP(a_i)\cup \{P(a_i)\}$. We also denote its separator, i.e., the set of ancestors that are constrained with $a_i$ or its descendants, as $Sep(a_i)$ \cite{petcu2006odpop}. Fig. 1 (b) gives a possible pseudo tree deriving from Fig. 1 (a). 
\subsection{DPOP and ADPOP}
DPOP is an important inference-based algorithm that performs dynamic programming on a pseudo tree, starting with a phase of utility propagation. In the phase, each agent joins the received utilities from its children with its local utility, eliminates its dimension by calculating the optimal utility for each assignment combination of its separator, and propagates the reduced utility to its parent. After that, a value propagation phase starts from the root agent. In the phase, each agent chooses the optimal assignment according to the utilities calculated in the previous phase and the assignments from its parent, and broadcasts the extended assignments to its children. The algorithm terminates when all agents have chosen their optimal assignments.

Although DPOP only requires a linear number of messages to solve a DCOP, its memory consumption is still exponential in induced width, which prohibits it from solving more complex problems. Thus, Petcu et. al, proposed ADPOP which is an approximate version of DPOP and allows the desired trade-off between solution quality and computational complexity. Specifically, ADPOP imposes a limit $maxDim$ on the maximum number of dimensions in each message. When the number of dimensions in an outgoing message exceeds the limit, the algorithm drops a set of dimensions to stay below the limit. That is, the algorithm computes an upper bound and a lower bound by applying a maximal/minimal projection on these dimensions. During the value propagation phase, agents can make decisions according to the highest utilities in either upper bounds or lower bounds.
\section{Proposed Method}
In this section, we present our proposed PT-ISABB, a two-phase hybrid complete algorithm for ADCOPs. We begin with a motivation, and then present the details of the inference phase and the search phase, respectively.
\subsection{Motivation}
The existing search-based complete algorithms for ADCOPs can only use \textit{local} knowledge to compute lower bounds, which leads to inefficient pruning. More specifically, unassigned agents report the best \textit{local} costs under the given partial assignments to compute lower bounds.  Taking Fig. 1 as an example, in AsymPT-FB agent $a_1$ can receive \textbf{LB\_report}s from $a_2$ and $a_3$. As a consequence, $a_1$ can only be aware of the lower bounds of $f_{21}$ and $f_{31}$ and does not have any knowledge about the remaining constraints (i.e., the constraints between $a_2$ and $a_3$, $a_2$ and $a_4$). On the other hand, inference algorithms like DPOP are able to aggregate and propagate the \textit{global} utility, but they are not applicable to ADCOPs due to a privacy concern. For example, $a_3$ needs to know both $f_{13}$ and $f_{23}$ to optimally eliminate $x_3$, which violates the privacies of $a_1$ and $a_2$. Thus, to overcome the pathologies, we propose a novel hybrid scheme to solve ADCOPs, which combines both inference and search. Specifically, the scheme consists of the following phases.
\begin{itemize}
	\item \textbf{Inference phase}: performing a bottom-up utility propagation with respect to a subset of constraints to build look-up tables for lower bounds
	\item \textbf{Search phase}: using a tree-based complete search algorithm for ADCOPs to exhaust the search space and guarantee optimality
\end{itemize}

In this paper, we propose a tailored version of ADPOP for the inference phase to avoid the severe privacy loss and exponential memory consumption in DPOP. Furthermore, we implement SABB on a pseudo tree for the search phase and propose an algorithm called PT-ISABB. Although they both operate on pseudo trees, our algorithm excels AsymPT-FB twofold. When its maximal dimension limit $k=\infty$, the lower bounds in our algorithm are at least as tight as the ones in AsymPT-FB (see Property 4.1 for detail). Moreover, PT-ISABB avoids to perform forward bounding which is expensive during the search phase.
\subsection{Inference Phase}
\begin{figure}
	\centering
	\includegraphics[scale=.9]{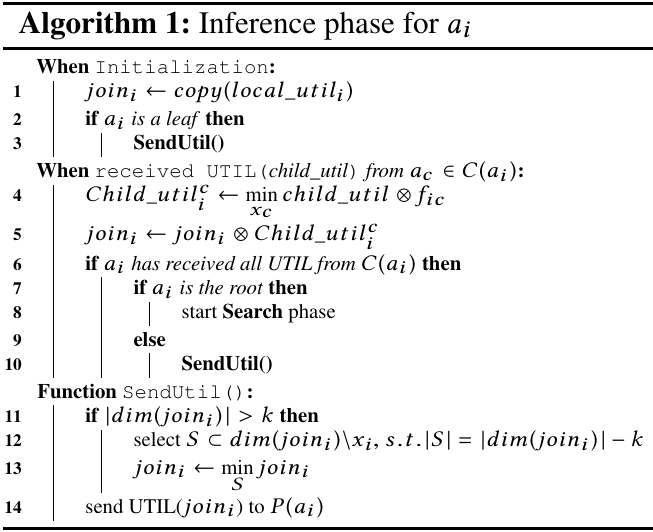}
	\caption{Pseudo code of inference phase}
\end{figure}
Fig. 2 gives the sketch of the inference phase for PT-ISABB. The phase begins with leaf agents who send their local utilities to their parents via UTIL messages (line 1 - 3). Particularly, if the number of dimensions in the utility exceeds the limit $k$ (line 11), we drop the dimensions of the highest ancestors to stay below the limit by a minimal projection (line 12 - 13). Here, $local\_util_i$ denotes the combination of the constraints between agent $a_i$ and its parent and pseudo parents enforced in $a_i$ side, i.e., 
$$local\_util_i=\bigotimes_{a_j\in AP(a_i)}f_{ij}$$
Note that in our algorithm we do not require the parent agents to disclose their private functions to perform inference exactly. In this way, a local utility table only involves the functions of that agent and the privacies of its parent and pseudo parents are therefore guaranteed. On the other hand, however, ignoring the private functions of parents and pseudo parents leads to severe inconsistencies when performing variable elimination. In other words, we actually trade lower bound tightness for privacy. We try to alleviate the problem by performing \textit{non-local} elimination which is elaborated as follows.

When $a_i$ receives a UTIL message from its child $a_c$, it joins the utility from $a_c$ with its corresponding private constraint function and then eliminates the dimension $x_c$ for a more complete utility $Child\_util_i^c$ (line 4). Compared to DPOP and ADPOP, the elimination of each variable is postponed to its parent in the pseudo-tree. Taking Fig. 1 (b) for example, the UTIL message from $a_3$ to $a_2$ is given by $$f_{32}+f_{31}$$ if the maximal dimension limit $k\ge 3$, and the elimination of $x_3$ is actually performed by $a_2$. That is, $$Child\_util_2^3=\min_{x_3}(f_{23}+f_{32}+f_{31})$$ Then, $a_i$ initiates the search phase after receiving all UTIL messages from its children if it is the root agent (line 6 - 8). Otherwise, it propagates the joint utility to its parent (line 9 - 10). 
\subsection{Search Phase}
\begin{figure}
	\centering
	\includegraphics[scale=.9]{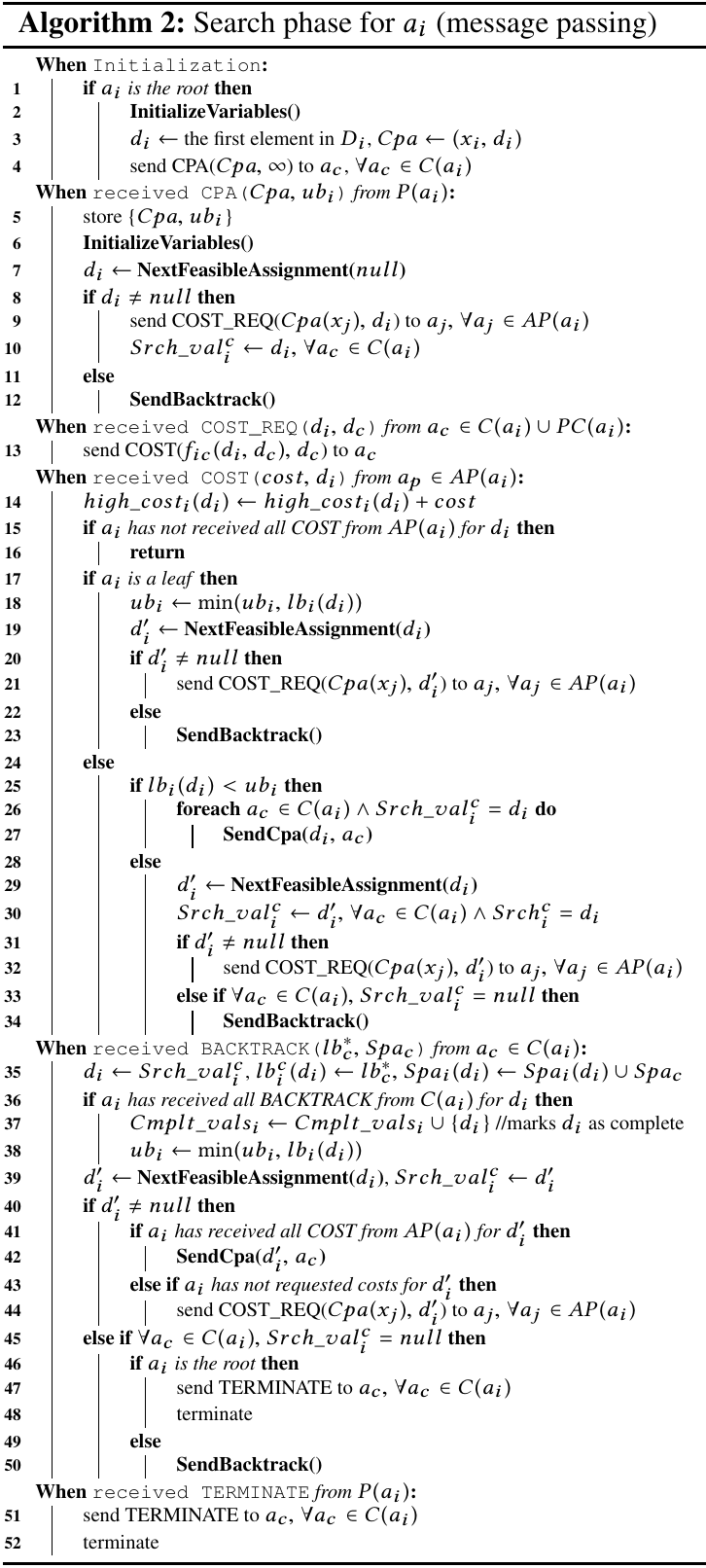}
	\caption{Pseudo code of search phase (message passing)}
\end{figure}
\begin{figure}
	\centering
	\includegraphics[scale=.9]{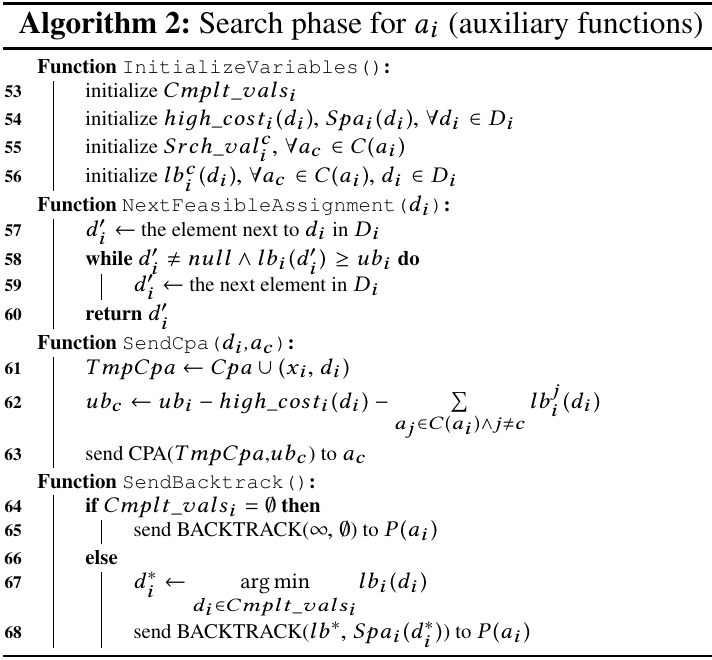}
	\caption{Pseudo code of search phase (auxiliary functions)}
\end{figure}
The phase performs a branch-and-bound search on a pseudo tree to exhaust the search space. Specifically,  each branching agent decomposes the problem into several subproblems, and each of its children solves a subproblem in parallel. To detect and discard the suboptimal solution, each agent maintains a upper bound for its subproblem and an lower bound for each value in its domain. Therefore, each agent $a_i$ needs to maintain the following data structures.
\begin{itemize}
	\item $Srch\_val_i^c$ records the assignment currently being explored in the subtree rooted at $a_c\in C(a_i)$. The data structure is necessary because asynchronous search is carried out in parallel in sub-trees based on different possible values of $x_i$.
	\item $high\_cost_i(d_i)$ is the cost for $d_i\in D_i$ between $a_i$ and its parent and pseudo parents under the current partial assignment ($Cpa$), which is initially set to the cost enforced in $a_i$ side. That is, 
	\begin{equation}
	\small
	high\_cost_i(d_i)=\sum_{a_j\in AP(a_i)}f_{ij}(d_i,Cpa(x_j))
	\end{equation}
	\item $lb_i^c(d_i)$ is the lower bound of child $a_c\in C(a_i)$ for $d_i\in D_i$, which is initially set to the utility under $Cpa$ and $d_i$. That is,
	\begin{equation}
	\small
	lb^c_i(d_i)=Child\_util_i^c(Cpa_{[Sep(a_c)]},x_i=d_i)
	\end{equation}
	where $Cpa_{[Sep(a_c)]}$ is a slice to $Cpa$ under $Sep(a_c)$, i.e., 
	$$Cpa_{[Sep(a_c)]}=\{(x_j,d_j)|(x_j,d_j)\in Cpa\land x_j\in Sep(a_c)\}$$
	When $a_i$ receives a BACKTRACK message from $a_c$, it replaces the initial lower bound with the actual cost reported by $a_c$ (or $\infty$ if $d_i$ is infeasible for $a_c$ given $Cpa$).
	\item $lb_i(d_i)$ is the lower bound for $d_i\in D_i$, i.e.,
	\begin{equation}
	\small
	lb_i(d_i)=high\_cost_i(d_i)+\sum_{a_c\in C(a_i)}lb_i^c(d_i)
	\end{equation}
	\item $Cmplt\_vals_i$ is the set of assignments for which $a_i$ has received all BACKTRACK messages from its children, and is initially set to $\emptyset$. Particularly, $Cmplt\_vals_i=D_i$ if $a_i$ is a leaf agent.
	\item $lb_i^*$ is the best cost explored under $Cpa$, which is given by
	\begin{equation}
	\small
	lb_i^*=\min_{d_i\in Cmplt\_vals_i}lb_i(d_i)
	\end{equation}
	Particularly, if $Cmplt\_vals_i=\emptyset$, $lb^*_i=\infty$.
	\item $Spa_i(d_i)$ is the optimal assignment to its subproblem under $Cpa$ when $x_i=d_i$ and is initially set to $\{x_i=d_i\}$. Particularly, $Spa_i(d_i^*)$ is the optimal solution if $a_i$ is  the root agent, where $d_i^*=\arg\min_{d_i\in Cmplt\_val_i}lb_i(d_i)$.
\end{itemize}

Fig. 3 and Fig. 4 give the pseudo codes of the search phase for PT-ISABB. The phase begins with the root agent sending the first element in its domain to its children (line 1 - 4). When an agent $a_i$ receives a CPA message from its parent, it first stores the partial assignment $Cpa$ and upper bound $ub_i$ and then finds the first feasible assignment (line 5 - 7), i.e., the first assignment $d_i$ such that $lb_i(d_i)<ub_i$ (line 57 - 60). If such an assignment exists, $a_i$ sends COST\_REQ messages to its parent and pseudo parents to request the private costs of other side for $d_i$ (line 8 - 10, line 13). Otherwise, it sends a BACKTRACK message with an infinity cost and an empty subproblem assignment (line 64 - 65) to its parent to announce that the given $Cpa$ is infeasible (line 11 - 12).

When $a_i$ receives a COST message for $d_i$, it adds the other side cost to $high\_cost_i(d_i)$ (line 14). After receiving all the COST messages for $d_i$ from its parent and pseudo parents, $a_i$ is able to determine whether it should continue to explore $d_i$. If $a_i$ is a leaf agent, it just updates the current upper bound and switches to the next feasible assignment $d^\prime_i$ (line 17 - 19) since the search space no longer needs to be expanded. If such $d^\prime_i$ exists, $a_i$ requests costs for the assignment (line 20 - 21). Otherwise, it backtracks to its parent by reporting the best cost and the best subproblem assignment explored under $Cpa$ (line 22 - 23, line 66 - 68). If $a_i$ is not a leaf agent and the current lower bound for $d_i$ is still less than its upper bound, it expands the search space by sending CPA messages to its children who are going to explore $d_i$ (line 25 - 27). Each message contains an extended partial assignment (line 61) and an upper bound which is the remainder after deducting the cost incurred by $d_i$ and the lower bounds of the other children from $a_i$'s upper bound (line 62). Otherwise, $d_i$ is proven to be suboptimal and the agent switches to the next feasible assignment (line 28 - 30). If such an assignment exists, $a_i$ requests costs for it (line 31 - 32). A backtrack takes place if all children exhaust $a_i$'s domain (line 33 - 34). 

When $a_i$ receives a BACKTRACK message for $d_i$ from a child $a_c$, it updates the corresponding lower bound $lb_i^c(d_i)$ with the actual cost $lb^*_c$  reported by $a_c$ if $Cpa$ and the assignment $d_i$ is feasible (otherwise $lb^*_c=\infty$), and merges the best assignments from $a_c$ (line 35). If $a_i$ has received all the BACKTRACK messages for $d_i$ from its children, it marks $d_i$ as complete and updates the current upper bound for its subproblem (line 36 - 38). $a_i$ also needs to determine the next assignment $d^\prime_i$ for $a_c$ to explore (line 39). If $d^\prime_i$ exists and $a_i$ has received all the COST messages from its parent and pseudo parents, it informs $a_c$ to explore $d^\prime_i$ by sending a CPA message (line 40 - 42). Otherwise, $a_i$ requests costs for $d^\prime_i$ if it has not been done (line 43 - 44). If $d^\prime_i$ does not exist and all children have exhausted $a_i$'s domain, $a_i$ informs its children to terminate and terminates itself if it is the root agent (line 45 - 48). Otherwise, it backtracks to its parent (line 49 - 50).
\section{Theoretical Results}
\subsection{Correctness}
\newtheorem{property}{Property}[section]
In this section, we first prove the termination and optimality, and further establish the completeness of PT-ISABB.  
\begin{lemma}
	PT-ISABB will terminate after a finite number of iterations.
\end{lemma}
\begin{proof}
	Directly from the pseudo codes, the inference phase will terminate since it only needs a linear number of messages. Thus, to prove the termination, it is enough to show that the same partial assignment cannot be explored twice in the search phase, i.e., an agent will not receive two identical $Cpa$s. Obviously, the claim holds for the root agent since it does not receive any CPA message. For an agent and a given $Cpa$ from its parent, it will send several CPA messages to each child. Since each of them contains the different assignments of the agent (line 29, line 39, line 57 - 60), the $Cpa$s sent to the child are all different. Therefore, the termination is hereby guaranteed. 
\end{proof}

\begin{lemma}
	For an agent $a_i$ and a given $Cpa$, the cost incurred by any assignment $Spa_i$ to the subtree rooted at $a_i$ with the assignment $(x_i=d_i)$ is no less than the corresponding lower bound $lb_i(d_i)$.
\end{lemma}
\begin{proof}
	The lemma for a leaf is trivial since $lb_i(d_i)$ is set to the cost of $d_i$ which is obviously no greater than the cost of the feasible assignment. We now focus on no-leaf agents.
	Recall that $a_i$ will replace the original lower bound with the actual cost reported by $a_c$ after receiving a BACKTRACK message for $d_i$ from $a_c\in C(a_i)$ (line 35). Thus, to prove the lemma, it is sufficient to show that the initial lower bound $lb_i^c(d_i)$ is no greater than the actual cost of $Spa_c,\forall a_c\in C(a_i)$, where $Spa_c\subset Spa_i$ is the assignment to the subtree rooted at $a_c$.
	
	Consider the induction basis, i.e., $a_i$'s children are leafs. For each child $a_c\in C(a_i)$, we have
	$$
	\small
	\begin{aligned}
	cost(Spa_c)&=\sum_{a_j\in AP(a_c)}f_{jc}(d_j,d_c)+f_{cj}(d_c,d_j)\\
	&\ge \min_{x_c}f_{ic}(d_i,x_c)+\sum_{a_j\in AP(a_c)}f_{cj}(x_c,d_j)\\
	&\ge Child\_util_i^c(Cpa_{[Sep(a_c)]},x_i=d_i)=lb_i^c(d_i)
	\end{aligned}
	$$
	where $d_l$ is the assignment to $x_l$ in $Cpa$ or $Spa_c$.The equation in the second to the last step holds when the maximal dimension limit $k=\infty$. Thus, the lemma holds for the basis.
	
	Assume that the lemma holds for all $a_c\in C(a_i)$. Next, we are going to show the lemma holds for $a_i$ as well. For each child $a_c\in C(a_i)$, we have
	$$
	\tiny
	\begin{aligned}
	cost(Spa_c)&=\sum_{a_j\in AP(a_c)}f_{cj}(d_c,d_j)+f_{jc}(d_j,d_c)+\sum_{a_{c^\prime}\in C(a_c)}cost(Spa_{c^\prime})\\
	&\ge\sum_{a_j\in AP(a_c)}f_{cj}(d_c,d_j)+f_{jc}(d_j,d_c)+\sum_{a_{c^\prime}\in C(a_c)}lb_c^{c^\prime}(d_c)\\
	&\ge\min_{x_c}f_{ic}(d_i,x_c)+\sum_{a_j\in AP(a_c)}f_{cj}(x_c,d_j)+\sum_{a_{c^\prime}\in C(a_c)}lb_c^{c^\prime}(x_c)\\
	&\ge Child\_util_i^c(Cpa_{[Sep(a_c)]},x_i=d_i)=lb_i^c(d_i)
	\end{aligned}
	$$
	which establishes the lemma.
\end{proof}
\begin{lemma}
	For an agent $a_i$ and a given $Cpa$, any assignment to the subtree rooted at $a_i$ with cost greater than $ub_i$ cannot be a part of a solution with cost less than the global upper bound.
\end{lemma}
\begin{proof}
	We will prove recursively, by showing that for a partial assignment $Spa_i$ to the subtree rooted at $a_i$ with $cost(Spa_i)>ub_i$, any partial assignment $Spa_j\supset Spa_i$ to the subtree rooted at $a_j$ will have $cost(Spa_j)>ub_j$ where $a_j=P(a_i)$. Note that $ub_i$ could be either an upper bound from $a_j$ via a CPA message (line 5) or a result of updating the upper bound locally (line 18, line 38). $a_i$ cannot backtrack by reporting $Spa_i$ in the latter case since there must exist a better partial assignment whose cost is $ub_i$. If $ub_i$ is received from $a_j$, according to line 62, we have
	$$ub_j=ub_i+high\_cost_j(d_j)+\sum_{a_c\in C(a_j)\land c\ne i}lb_j^c(d_j)$$
	Thus, $cost(Spa_i)>ub_i$ necessarily means that any partial assignment $Spa_j\supset Spa_i$ will have $cost(Spa_j)>ub_j$.
\end{proof}
\begin{theorem}
	PT-ISABB is complete.
\end{theorem}
\begin{proof}
	Immediately from Lemma 4.1, Lemma 4.2 and Lemma 4.3, the algorithm will terminate and all pruned assignments are suboptimal. Thus, PT-ISABB is complete.
\end{proof}
\subsection{Lower bound tightness}
\begin{property}
	For an agent $a_i$ and a given $Cpa$, the initial lower bound $lb_i^c(d_i)$ of $a_c\in C(a_i)$ for $d_i$ is at least as tight as the one in AsymPT-FB when the maximal dimension limit $k=\infty$.
\end{property}
\begin{proof}
	In AsymPT-FB, the lower bound for $a_c$ after receiving all the \textbf{LB\_Report}s from the subtree rooted at $a_c$ is given by the sum of the best single side local costs of $a_c$'s descendants under $Cpa$. That is,
	$$
	\scriptsize
	\begin{aligned}
	SubtreeLB_i^c(d_i)&=\sum_{a_j\in Desc(a_c)}\min_{x_j}\sum_{a_l\in Sep(a_c)\cap PP(a_j)}f_{jl}(x_j,d_l)\\
	&+\min_{x_c}\sum_{a_l\in AP(a_c)}f_{cl}(x_c,d_l)
	\end{aligned}
	$$
	where $Desc(a_c)$ is the set of the descendants of $a_c$. For the sake of clarity, we denote the vector of $x_c$ and its descendant variables as $\textbf{x}_\textbf{c}$. Next, we will show $lb_i^c(d_i)\ge SubtreeLB_i^c(d_i)$. Since $k=\infty$, the inference phase does not drop any dimension. Thus, we have
	$$
	\scriptsize
	\begin{aligned}
	lb_i^c(d_i)&=\min_{\textbf{x}_\textbf{c}}\sum_{a_j\in Desc(a_c)}\left(\sum_{a_l\in AP(a_j)\cap Sep(a_c)}f_{jl}(x_j,d_l)\right.\\
	&\left.+\sum_{a_l\in AP(a_j)\cap Desc(a_i)}f_{jl}(x_j,x_l)+\sum_{a_l\in C(a_j)}f_{jl}(x_j,x_l)\right)\\
	&+\sum_{a_l\in AP(a_c)}f_{cl}(x_c,d_l)+\sum_{a_l\in C(a_c)}f_{cl}(x_c,x_l)+f_{ic}(d_i,x_c)\\
	&\ge\min_{\textbf{x}_\textbf{c}}\sum_{a_j\in Desc(a_c)}\sum_{a_l\in AP(a_j)\cap Sep(a_c)}f_{jl}(x_j,d_l)\\
	&+\sum_{a_l\in AP(a_c)}f_{cl}(x_c,d_l)
	\end{aligned}
	$$
	Since $x_c\in \textbf{x}_\textbf{c}$ and $AP(a_j)\supset PP(a_j)$, the right-hand side of the inequality in the last step can be further reduced. That is,
	$$
	\scriptsize
	\begin{aligned}
	lb_i^c(d_i)&\ge\min_{\textbf{x}_\textbf{c}\backslash x_c}\sum_{a_j \in Desc(a_c)}\sum_{a_l\in AP(a_j)\cap Sep(a_c)}f_{jl}(x_j,d_l)\\
	&+\min_{x_c}\sum_{a_l\in AP(a_c)}f_{cl}(x_c,d_l)\\
	&\ge \min_{\textbf{x}_\textbf{c}\backslash x_c}\sum_{a_j\in Desc(a_c)}\sum_{a_l\in PP(a_j)\cap Sep(a_c)}f_{jl}(x_j,d_l)\\
	&+\min_{x_c}\sum_{a_l\in AP(a_c)}f_{cl}(x_c,d_l)\\
	&\ge\sum_{a_j\in Desc(a_c)}\min_{x_j}\sum_{a_l\in Sep(a_c)\cap PP(a_j)}f_{jl}(x_j,d_l)\\
	&+\min_{x_c}\sum_{a_l\in AP(a_c)}f_{cl}(x_c,d_l)\\
	&=SubtreeLB_i^c(d_i)
	\end{aligned}
	$$
	which concludes the property.
\end{proof}
\subsection{Complexity}
Since an agent $a_i$ stores $Child\_util_i^c$ and $lb_i^c(d_i)$ for each child, the overall space complexity in the worst case (i.e., $k=\infty$) is $O(|C(a_i)|d_{max}^{|Sep(a_i)|+1}+|C(a_i)||D_i|)$ where $d_{max}=\max_{a_j\in Sep(a_i)}|D_j|$. Since it contains all the dimensions of $Sep(a_i)$ and itself, the size of an UTIL message from $a_i$ is $O(d_{max}^{|Sep(a_i)|+1})$ when $k=\infty$. For a CPA message, it consists of the assignment of each agent and an upper bound. Thus, the size of a CPA message is $O(|A|)$. Other messages including COST\_REQ, COST, BACKTRACK and TERMINATE carry several scalars and thus they only require $O(1)$ space.

Different than standard DPOP/ADPOP, PT-ISABB only requires $|A|-1$ messages in the inference phase since it does not have the value propagation phase. Like any other search based complete algorithm, the message number of the search phase grows exponentially with respect to the agent number.

\section{Experimental Results}
We empirically evaluate PT-ISABB with state-of-the-art search-based complete algorithms for ADCOPs including SABB, ATWB and AysmPT-FB on three configurations. To demonstrate the real power of non-local elimination, we also consider SABB on a pseudo-tree (PT-SABB) and the local elimination version of PT-ISABB (PT-ISABB, local) with $k=\infty$. In the first ADCOP configuration, we set the graph density to 0.25, the domain size to 3 and vary the agent number from 8 to 18. The second configuration is ADCOPs with 8 agents and the domain size of 8. The graph density varies from 0.25 to 1. In the last configuration, we consider asymmetric MaxDCSPs with 10 agent, the domain size of 10 and the graph density of 0.4, and the tightness varies from 0.1 to 0.8. For each of the settings, we generate 50 random instances and the results are averaged over all instances. In our experiments, we use the number of non-concurrent logical operations (NCLO) \cite{Netzer2012} to evaluate hardware-independent runtime, in which the logical operations in the inference phase are accesses to utility tables, and for the search phase and other competitors they are constraint checks. Also, we use the message number and the size of total information exchanged to measure the network load. For asymmetric MaxDCSPs, we use entropy \cite{Brito2009Distributed} to quantify the privacy loss \cite{litov2017forward,grinshpounGZNM13}. The experiments are conducted on an i7-7820x workstation with 32GB of memory and for each algorithm we set the timeout to 2 minutes. 
\begin{figure}
	\centering
	\includegraphics[scale=.38]{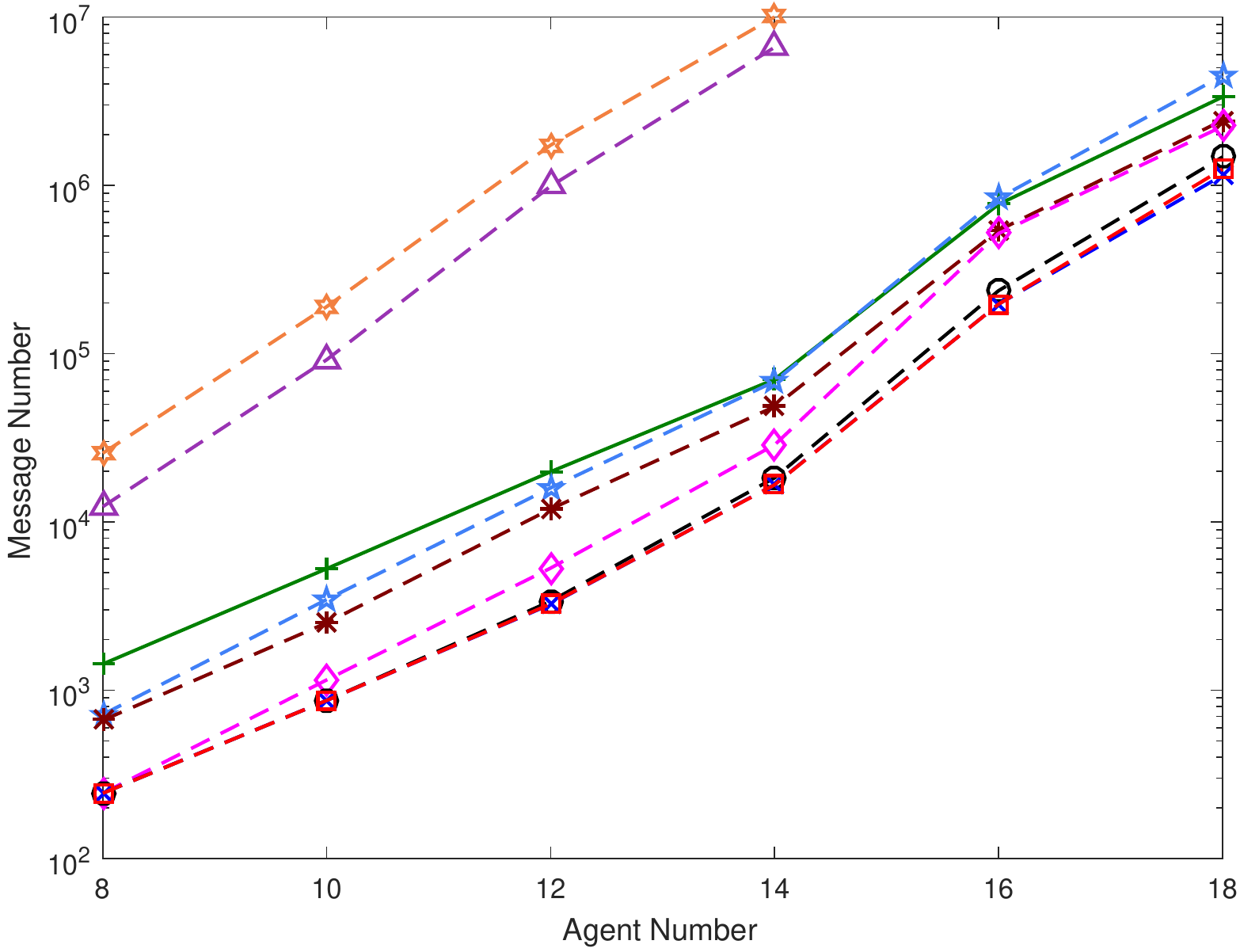}\\
	(a)\\
	\includegraphics[scale=.38]{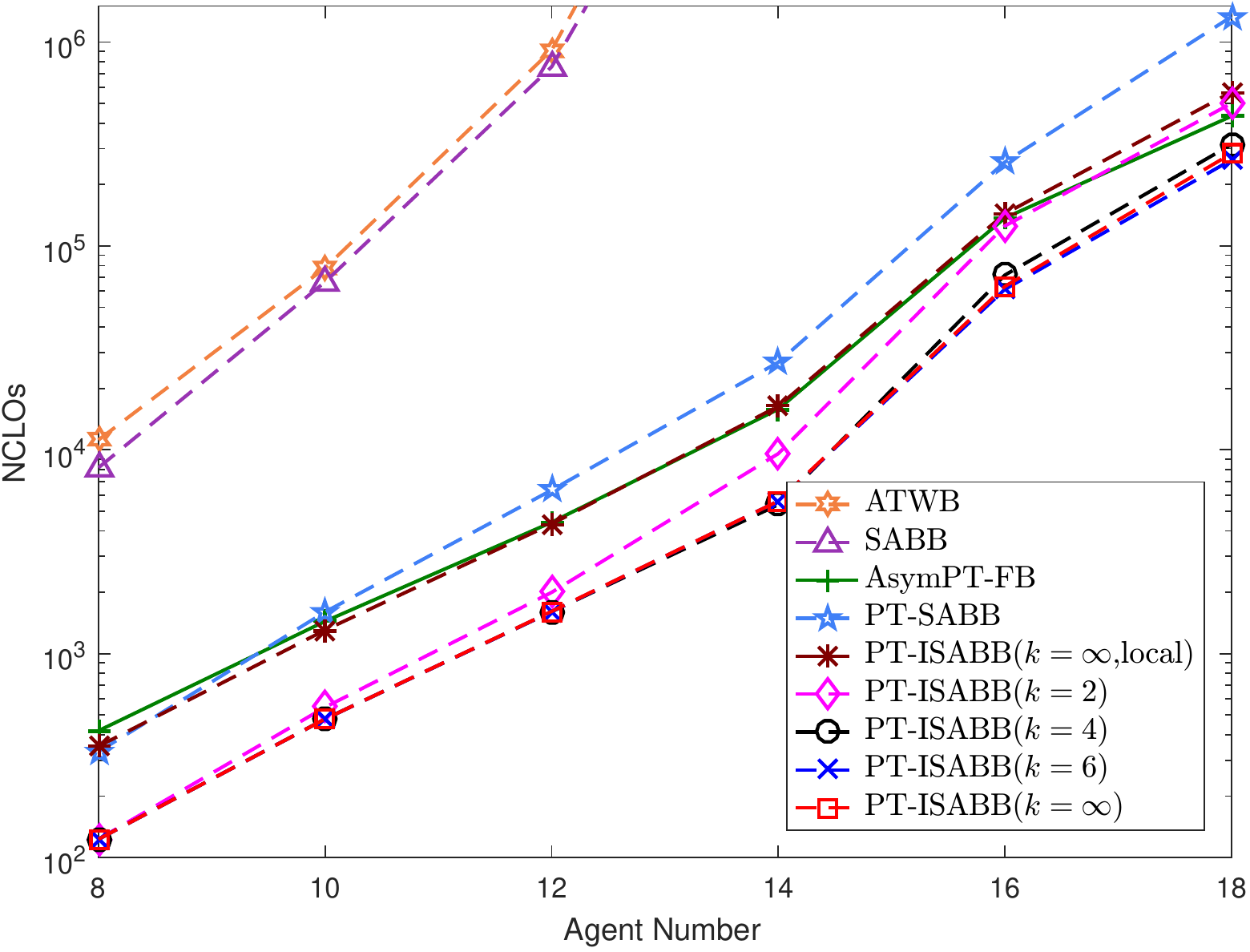}\\
	(b)
	\caption{Performance comparison under different agent numbers}
\end{figure}
\begin{figure}
	\centering
	\includegraphics[scale=.38]{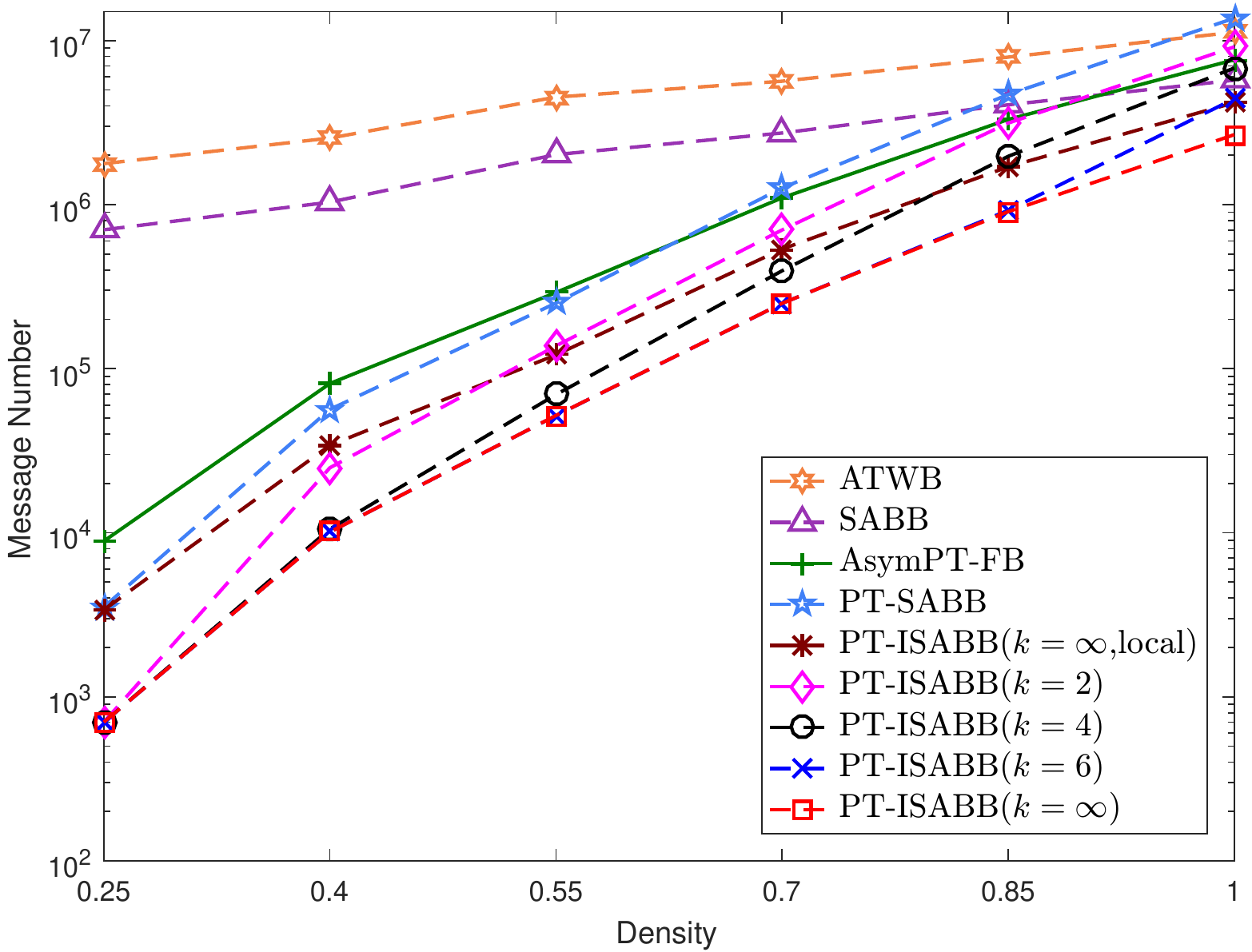}\\
	(a)\\
	\includegraphics[scale=.38]{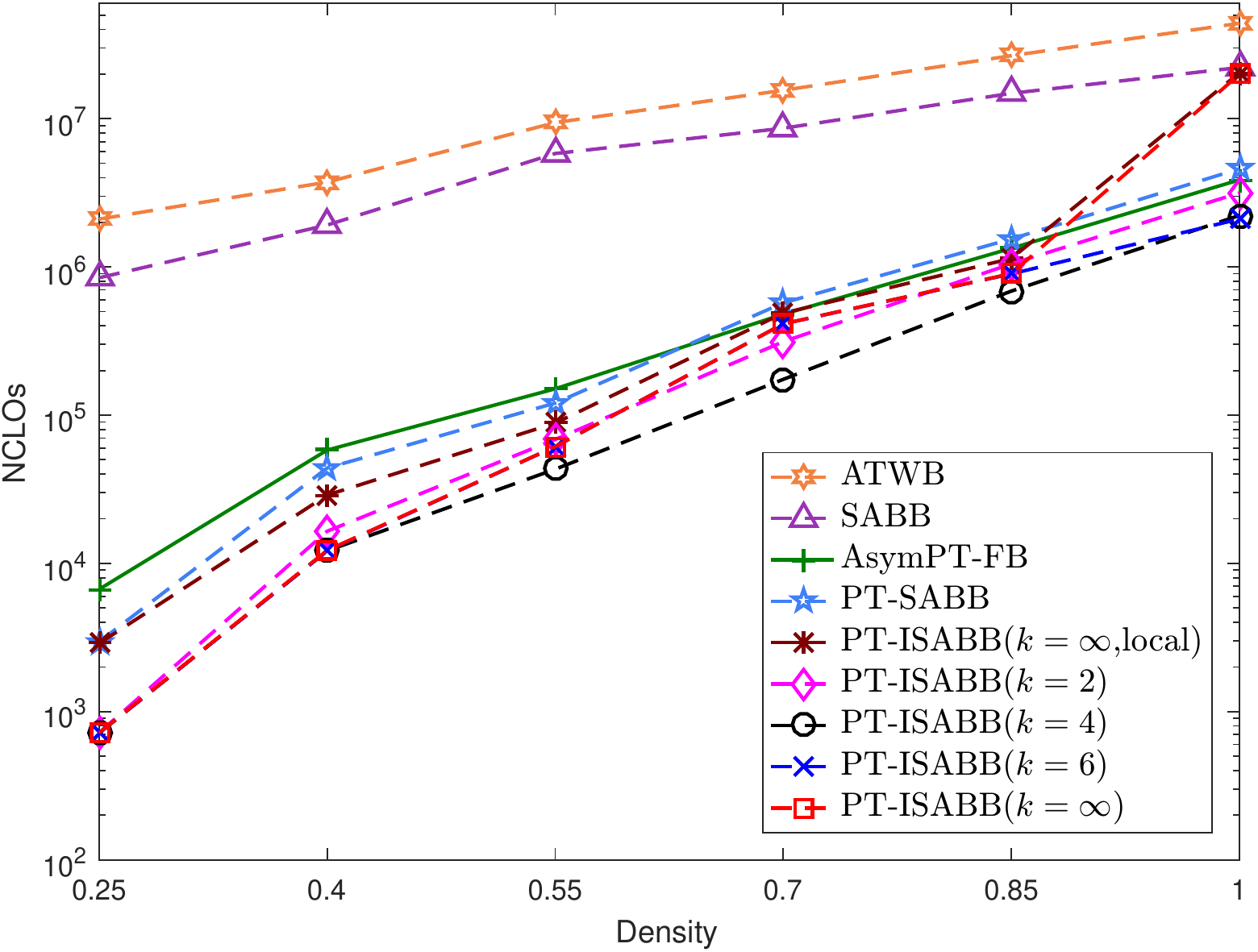}\\
	(b)
	\caption{Performance comparison under different graph densities}
\end{figure}

Fig. 5 gives the performance comparison on different agent numbers, and the average induced widths in the experiments are 1 $\sim$ 6.84. It can be seen from the figure that all the algorithms suffer from exponential overheads as the agent number grows. Among them, our proposed PT-ISABB requires significant fewer messages and lower NCLOs than the other competitors, which demonstrates the superiority of the hybrid execution of inference and search. On the other hand, although PT-ISABB ($k=\infty$, local) employs an complete one-side inference to construct the initial lower bounds, it is still inferior to PT-ISABB with $k>2$, which demonstrate the necessity of non-local elimination. Besides, it is worth noting that PT-ISABB requires much fewer messages than AsymPT-FB even when the maximal dimension limit $k$ is small (e.g., $k=2$). That is because PT-ISABB does not rely on forward bounding which is expensive in message-passing to compute lower bounds. Moreover, the phenomenon also indicates that our algorithm can produce tighter lower bounds even if the memory budget is relatively low.

Fig. 6 gives the results under different graph densities. The average induced widths here are 1 $\sim$ 6. Note that in this configuration, the size of the search space does not change and the complexity is reflected in the topologies. It can be concluded from the figure that all the tree-based algorithms exhibit great superiorities when the graph density is low, and the advantages vanish as the density grows. That is because those algorithms can effectively parallel the search processes on sparse problems. Dense problems, on the other hand, usually result in pseudo trees with low branching factors, making the tree-based algorithms require more messages than SABB. Even so, our proposed PT-ISABB with large $k$ still outperforms SABB when the problems are fully connected, which demonstrates the necessity of tighter lower bounds. Additionally, the figure also indicates that PT-ISABB with different $k$ performs similarly on sparse problems, but the performances vary a lot on dense problems. That is due to the fact that the induced width of a pseudo tree is relatively small when solving a spare problem and thus only a small set of dimensions is dropped during the inference phase. Besides, it can be seen from the figure that although both PT-ISABB ($k=\infty$, local) and PT-ISABB ($k=\infty$) perform complete inferences, the non-local elimination version requires lower NCLOs and fewer message numbers in most of the settings. That is because the non-local elimination version can provide tighter lower bounds which result in efficient pruning, and thus the algorithm incurs fewer constraint checks and messages in the search phase.
\begin{figure}
	\centering
	\includegraphics[scale=.38]{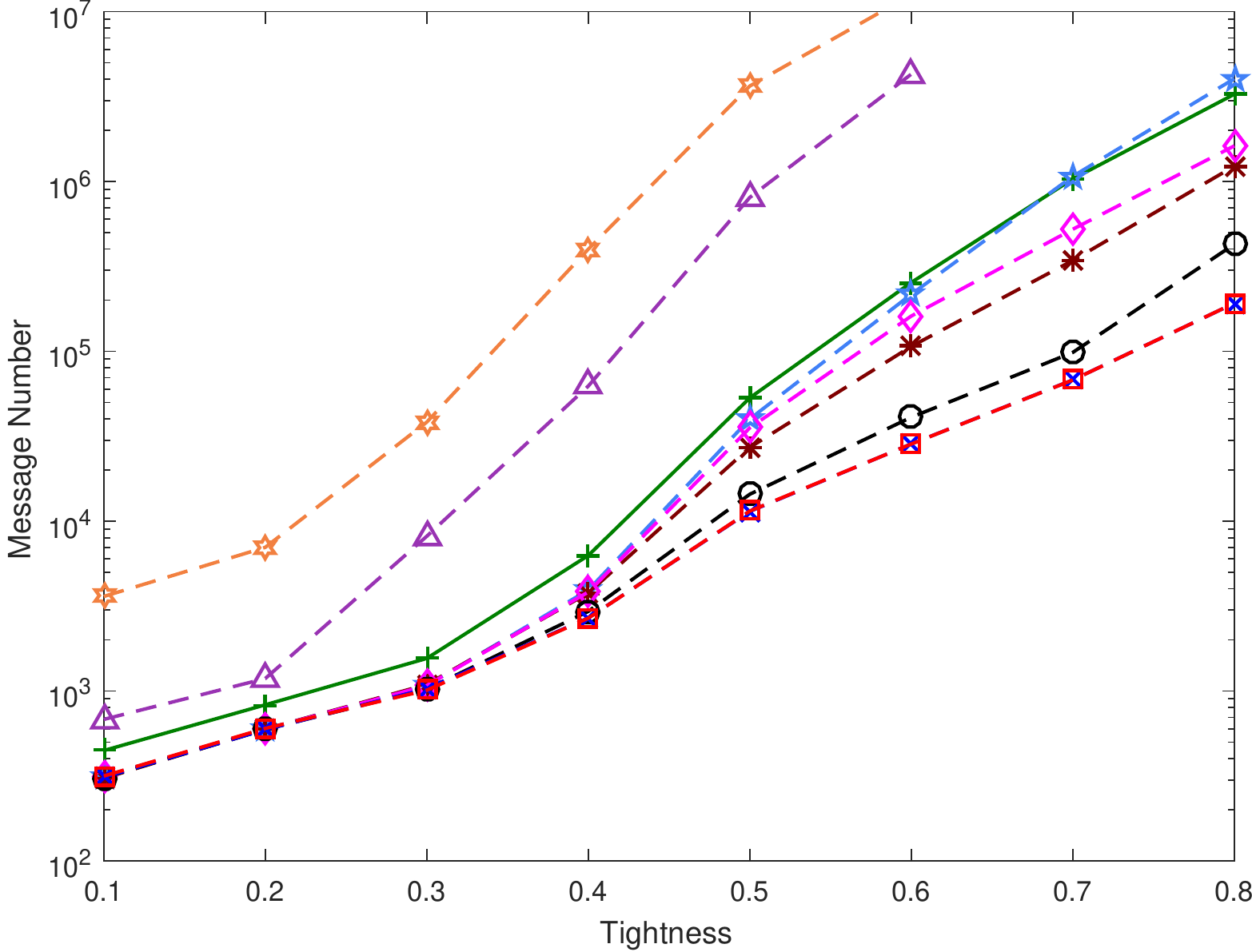}\\
	(a)\\
	\includegraphics[scale=.38]{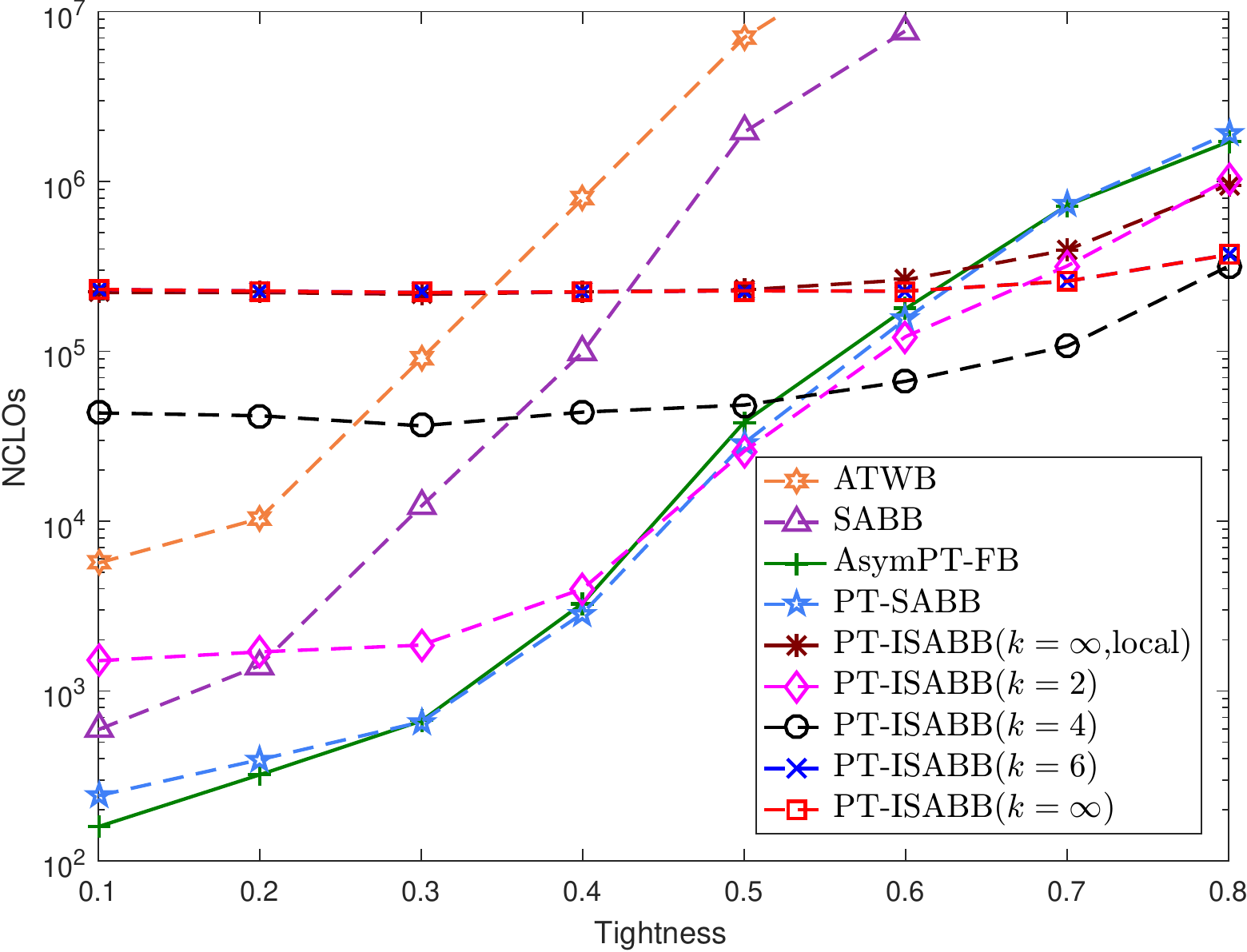}\\
	(b)
	\caption{Performance comparison under different tightness}
\end{figure}

\begin{table}
	\centering
	\includegraphics[scale=.52]{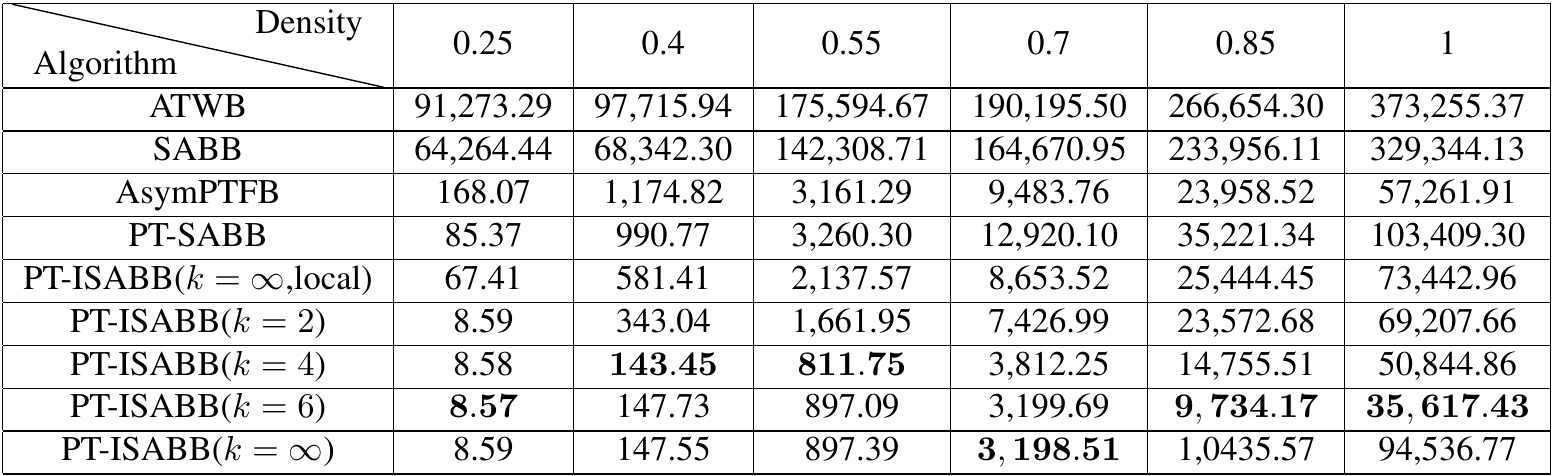}
	\caption{The size of total information exchanged of each algorithm under different graph densities (in KB)}
\end{table}

Table 1 presents the size of total information exchanged of each algorithm under different densities. It can be seen from the table that all the non-local elimination versions of PT-ISABB exhibit great advantages over the other search-based competitors, except that PT-ISABB ($k=2$) is slightly inferior to AsymPT-FB when solving the fully-connected problems. The phenomenon indicates that although a message in the inference phase is generally larger than the one in the search phase, the algorithms can still gain great reductions on network traffics since the effective pruning in the search phase greatly reduces the message number. Besides, it is interesting to find that a large dimension limit (e.g., $k=\infty$) does not necessarily result in the smallest traffic. That is because the size of a message in the inference phase is exponential to the minimum of the induced width and the dimension limit. Besides, it should be noted that although PT-SABB requires more messages than ATWB and SABB when solving fully-connected problems according to Fig. 6, it still incurs much smaller traffic due to the fact that the last agent in ATWB and SABB needs to broadcast the reached complete solution to all other agents once a new solution is constructed. In contrast, agents in PT-SABB only back up the assignments to their descendants, which are subsets of the complete solution, to their parents via BACKTRACK messages.

Fig. 7 presents the results on asymmetric MaxDCSPs with different tightness, and the average induced width is 3.92. This configuration neither increases the search space nor affects the topologies, but instead increases the difficulty of pruning. All the algorithms except ATWB produces few messages when solving problems with low tightness. That is because on these problems the algorithms can find low upper bounds very quickly to prune most of the search space. As the tightness grows, the number of prohibited combinations increases and the algorithms can no longer find low upper bounds promptly. As a result, the algorithms require much more search efforts to exhaust the search space. Since they cannot exploit topologies to accelerate the search process, SABB and ATWB perform poorly and can only solve the problems with tightness up to 0.6. On the other hand, the tree-based algorithms divide a problem to several smaller subproblems at each branching agent and search the subproblems in parallel. Thus, those algorithms exhibit better performances and solve all the problems. Among them, our proposed PT-ISABB with $k\ge 4$ incurs much smaller overheads, which demonstrates the effectiveness of the inference phase in computing tighter lower bounds. In other words, although PT-ISABB only guarantees to produce lower bounds as tight as the ones of AsymPT-FB when $k=\infty$ according to Property 4.1, it requires less memory consumption to compute such lower bounds in practice. Besides, it can be seen from the figure that PT-SABB incurs smaller communication overheads than AsymPT-FB when solving the problems with low tightness, which demonstrates forward bounding is expensive in message-passing. Additionally, it can be concluded that PT-ISABB algorithms with large $k$ require much more NCLOs than the other competitors when solving problems with low tightness. That is no surprise since inference on problems with large domain sizes would be more expensive, and a search-based algorithm actually can find a feasible solution very quickly even if the lower bounds are poor when solving these problems.
\begin{figure}
	\centering
	\includegraphics[scale=.38]{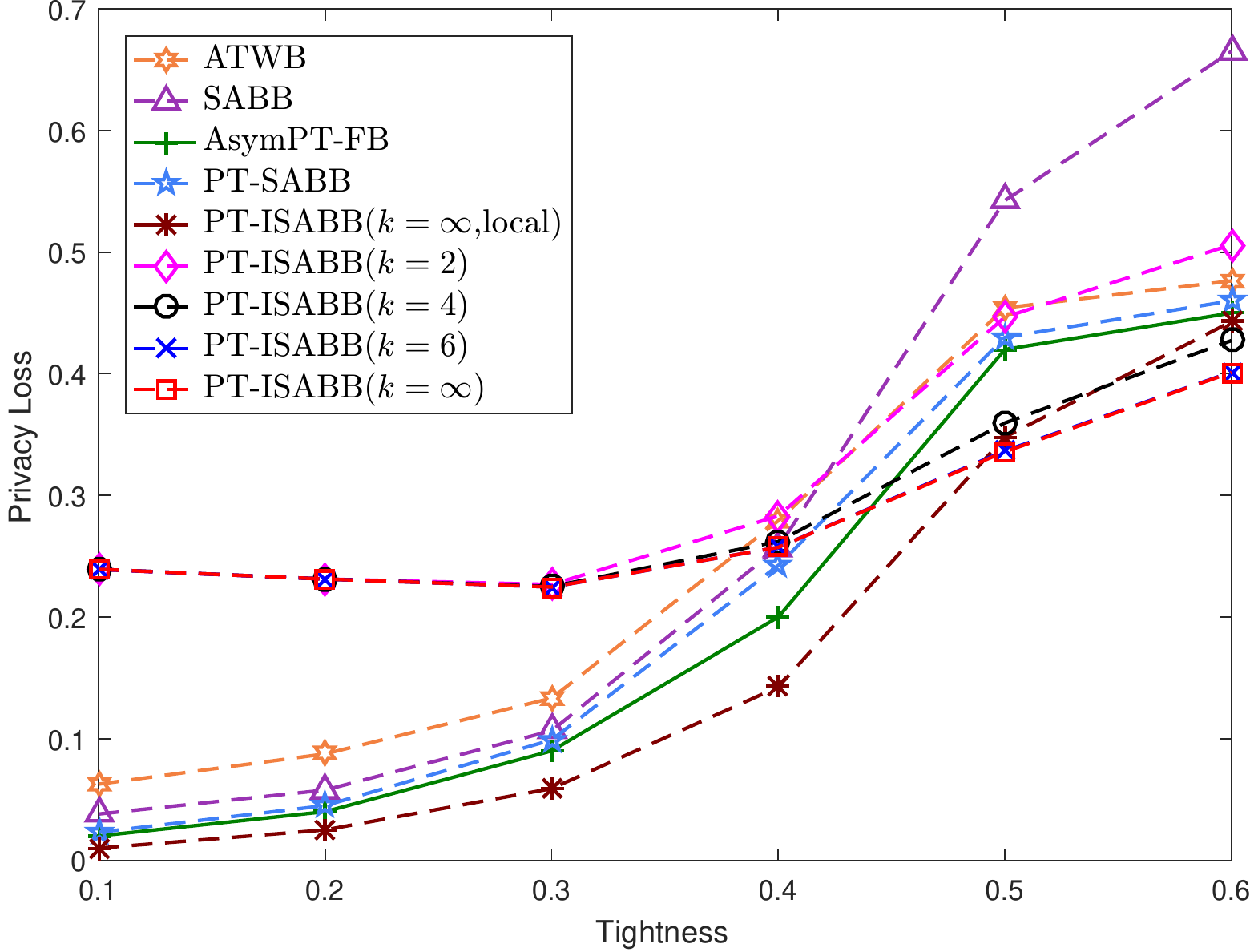}
	\caption{Privacy losses under different tightness}
\end{figure}

Fig. 8 gives privacy losses under different tightness. Privacy loss in PT-ISABB comes from both the inference phase and the search phase. Specifically, since the variable elimination is performed by parents (i.e., line 4 of the inference phase), parents can easily figure out which pairs of assignments are feasible with respect to the constraints enforced in children sides from the zero entries of the utilities from children. Thus, the inference phase would cause a half privacy loss only on each tree edge in the worst case, which is still better than leaking at least a half privacy if we directly use DPOP to solve problems. Besides, the direct disclosure mechanism of the search phase in which agents request their (pseudo) parents to expose the private costs before expanding the search space also leads to the privacy loss. However, the loss could be much reduced via the effective pruning. It can be seen that our proposed algorithm leaks more privacy than the other competitors when solving the problems with low tightness. That is no surprise because these problems usually have feasible solutions, which leads to the fact that most of entries in a utility from a child are zero. That being said, PT-ISABB with $k\ge 4$ leaks less privacy than the other competitors when solving the problems with high tightness. The reason is twofold: parents can no longer infer the feasible assignment pairs as the tightness grows, and the inference phase produces tight lower bounds which lead to the effective pruning in the search phase. Besides, it is worth mentioning that the local elimination version of PT-ISABB performs better in terms of privacy preservation when solving the problems with low tightness. That is because variables are already eliminated before sending utilities to their parents. As a result, parents can only know the best utilities they can achieve, but cannot figure out the corresponding assignments to their children.
\section{Conclusion}
It is known that DPOP/ADPOP for DCOP cannot be directly applied to ADCOP due to a privacy concern. In this paper, we take ADPOP into solving ADCOP for the first time by combining with a tree-based variation of SABB, and present a two-phase complete algorithm called PT-ISABB. In the inference phase, a non-local elimination version of ADPOP is performed to solve a subset of constraints and build look-up tables for the tighter lower bounds. In the search phase, a tree-based variation of SABB is implemented to exhaust the search space. The experimental results show that our algorithms exhibit great superiorities over state-of-the-art search based algorithms, as well as the local elimination version of PT-ISABB. Also, our algorithms leak less privacy when solving complex problems.
\begin{acks}
	The authors would like to thank the anonymous referees for
	their valuable comments and helpful suggestions. This work is
	supported by the \grantsponsor{}{Chongqing Research Program of Basic Research and Frontier Technology}{} under Grant
	No.:\grantnum{}{cstc2017jcyjAX0030}, \grantsponsor{}{Fundamental Research Funds for the Central Universities}{} under Grant
	No.: \grantnum{}{2018CDXYJSJ0026},
	\grantsponsor{}{National Natural Science Foundation of China}{} under Grant
	No.: \grantnum{}{51608070}
	and
	the \grantsponsor{}{Graduate Research and Innovation Foundation of Chongqing, China}{} under Grant
	No.: \grantnum{}{CYS18047}
\end{acks}


\bibliographystyle{ACM-Reference-Format}  
\balance  
\bibliography{ref.bib}  

\end{document}